\documentclass[aps,pra,twocolumn,superscriptaddress]{revtex4-1}
\usepackage{amsfonts}
\pdfoutput=1

\usepackage{graphicx,amsmath,amssymb,color}
\usepackage{hyperref} 
\usepackage{braket}

\newtheorem{proposition}{Proposition}
\newenvironment{proof}[1][Proof]{\begin{trivlist}
\item[\hskip \labelsep {\bfseries #1}]}{\end{trivlist}}
\begin{document}


\title{Geometric Multiaxial Representation of N-qubit Mixed Symmetric Separable States}



\author{Suma SP}
\email[]{sumarkr@gmail.com}
\author{Swarnamala Sirsi}
\affiliation{Yuvaraja's College, University of Mysore, Mysuru.}
\author{Subramanya Hegde}
\affiliation{School of Physics, Indian Institute of Science Education and Research, Thiruvananthapuram.}
\author{Karthik Bharath}
\affiliation{University of Nottingham, Nottingham, U.K.}


\begin{abstract}
Study of an N qubit mixed symmetric separable states is a long standing challenging problem as there exist no unique separability criterion. In this regard, we take up the N-qubit mixed symmetric separable states for a detailed study as these states are of experimental importance and offer elegant mathematical analysis since the dimension of the Hilbert space reduces from $2^{N}$ to $N+1$.  Since there exists a one to one correspondence between spin-j system and an N-qubit symmetric state, we employ Fano statistical tensor parameters for the parametrization of spin density matrix. Further, we use geometric multiaxial representation(MAR) of density matrix to characterize the mixed symmetric separable states. Since separability problem is NP hard, we choose to study it in the continuum limit where mixed symmetric separable states are characterized by the P-distribution function $\lambda(\theta,\phi)$. We show that the N-qubit mixed symmetric separable state can be visualized as a uniaxial system if the distribution function is independent of $\theta$ and $\phi$. We further choose distribution function to be the most general positive function on a sphere and observe that the statistical tensor parameters characterizing the N-qubit symmetric system are the expansion coefficients of the distribution function. As an example for the discrete case, we investigate the MAR of a uniformly weighted two qubit mixed symmetric separable state. We also observe that there exists a correspondence between separability and classicality of states.\\\\

\end{abstract}

\maketitle
 Study of separable states is the corner stone of entanglement problem. Bell inequalities were first used 
for the identification of entanglement. The most operationally convenient criterion for the detection of
entanglement is given by Peres-Horodecki and is called positive partial transpose (PPT)
criterion \cite{Peres} which is necessary and sufficient for $2 \times 2$
and $2 \times 3$ systems only\cite{horodecki2}. There exist other criteria in literature for detecting entanglement. One among them is
the realignment criterion \cite{chen}, which exhibits a
powerful PPT entanglement detection capability. Entanglement witness\cite{lewenstein, horodecki2} and uncertainty relations\cite{takeuchii} pose operational difficulty as they depend on the expectation value of some observables for the state in question.
 \\ In practice we deal with mixed states rather than pure states due to decoherence effects and hence it is of great importance to study  mixed separable states. There exists many important papers \cite{sp, dur, acin, eggeling, eltscka, gunhe, jung, gunhe2, huber} for mixed states in the literature; classification of local unitary equivalent classes
of symmetric N-qubit mixed states and an algorithm to identify pure separable states\cite{samin} based on the geometrical Multiaxial Representation
(MAR) of the density matrix\cite{rama} have been investigated. Makhlin\cite{maklin} has
presented a complete set of 18 local polynomial invariants of two qubit mixed states and
demonstrated the usefulness of these invariants to study entanglement. Also, detection of multipartite entanglement has been studied in depth(see for example \cite{horodecki3, yu, hiesmayr}). Geometric entanglement properties of pure symmetric N qubit states are studied in detail\cite{pgeometric}. To this day, no generally accepted theory exists for the classification and quantification of entanglement for mixed states. 
\\ A general N qubit mixed state resides in the Hilbert space of dimension $ 2^{N} \otimes 2^{N}$ which makes the mathematical computations complicated except for the lower N values whereas permutationally symmetric N qubit mixed states residing in $N+1$ dimensional Hilbert space not only offer elegant mathematical analysis but are also useful in variety of quantum information tasks. They occur naturally as ground states in some Bose- Hubbard models and are the most experimentally investigated states. But relatively a not much is explored about the entanglement or the separability criteria of these states and the investigation is mostly restricted to states like W and GHZ \cite{mermin, cabello, heaney}. Recently F. Bohnet-Waldra et.al, \cite{ppt} have studied positive-partial-transpose (PPT) separability criterion for symmetric states of multiqubit systems
in terms of matrix inequalities. They have also established a correspondence between classical spin states and symmetric separable states. Analytical expression for quantumness of  pure spin-1 state or equivalently two qubit pure symmetric state  using Majorana representation of density matrix is given in \cite{quantumness}. Further, this has been extended numerically to provide an upper bound of quantumness for mixed states. Majorana representation can not be extended naturally to study mixed symmetric states. Therefore in this paper we employ the little known geometric MAR of spin-j system to study the separability problem of mixed symmetric states. This method can also be used to investigate the quantumness of such states analytically. \\ \\
This paper is organized as follows: In section I we discuss the correspondence between symmetric states and spin systems. In Section II we explain the decomposition of density matrix in terms of the well known Fano statistical tensor parameters. Section III contains the description of the Multiaxial representation of pure and mixed density matrices. Section IV consist of two propositions which illustrate the conditions to be satisfied by  mixed symmetric separable density matrix. Section V deals with the Multiaxial representation of mixed symmetric separable states and their characterization.\\ 
\section{Correspondence between Symmetric states and Spin systems}
Set of N-qubit states that remains unchanged by permutation of individual particles are called symmetric states. That is, $\pi_{i,j} \rho^{symm}_{1,2....N} = \rho^{symm}_{1,2....N} \pi_{i,j} = \rho^{symm}_{1,2....N}$ where $\pi_{i,j}$ is called permutation operator, $i\neq j=1,2....N$. A general N-qubit state belongs to the Hilbert space  $C^{2^{\otimes N}}$ and  is represented by a density matrix of dimension  $2^{N} \times 2^{N}$. An N-qubit symmetric state has one-to-one correspondence with a spin-$j$ state where $j=\frac{N}{2}$. Therefore the $N+1$-dimensional symmetric subspace can be identified with a $2j+1$-dimensional Hilbert space which is the carrier space of the angular momentum operator $\bf{J}$. We focus on the such symmetric states in this article as they are of considerable interest. 

\section{Fano representation of spin-\lowercase{j} assembly}
A general spin-$j$ density matrix can be represented in terms of statistical tensor parameters \cite{Fano1953, Fano1957, Fano1983, Fano} $t^{k}_{q}\,{'}s$:
\begin{equation}
\rho(\vec{J}) = \frac {Tr(\rho)}{(2j+1)}\sum^{2j}_{k=0}\,\sum^{+k}_{q=-k}\,\, t^{k}_{q}\, \tau^{k^{\dagger}}_{q}(\vec{J})\,\,,
\end{equation}
where $ \vec{J} $ is the angular momentum operator with components $ J_{x}, J_{y}, J_{z}$. The operators $\tau^{k}_{q}$, (with $\tau^{0}_{0} = I$, the identity operator) are irreducible tensor operators of rank $k$ in the $ 2j+1$ dimensional angular momentum space with projection $q$ along the axis of quantization in $\mathbb{R}^3$. The elements of $\tau^{k}_{q}$ in the angular momentum basis $|jm\rangle, m=-j,\ldots,+j$ are given by 
$\langle{jm'}|\tau^{k}_{q}(\vec{J})|{jm}\rangle = [k]\,\,C(jkj;mqm')$, where $C(jkj;mqm')$ are the Clebsch-Gordan coefficients and $ [k]=\sqrt{2k+1} $. The $\tau^{k}_{q}$ satisfy the orthogonality relations
$$Tr({\tau^{k^{\dagger}}_{q}\tau^{k^{'}}_{q^{'}}})= (2j+1)\,\delta_{kk^{'}} \delta_{qq^{'}}, $$
where $ \tau^{k^{\dagger}}_{q}=(-1)^{q}\,\tau^{k}_{q}$
and 
$$t^{k}_{q} =  Tr(\rho\,\tau^{k}_{q})=\sum_{m=-j}^{+j}\,\langle{jm}|\rho\,\tau^{k}_{q}|{jm}\rangle.$$
Since $\rho$ is Hermitian and $\tau^{k^{\dagger}}_{q} = (-1)^{q}\tau^{k}_{-q}$,  the complex conjugates $t^{k}_{q}\,{'}s$ satisfy the condition $t^{k^{*}}_{q}=(-1)^{q}\,t^{k}_{-q}.$
 Furthermore, $\rho=\rho^{\dagger}$ and $Tr(\rho)=1$ imply that $\rho$ can be specified by    $n^{2}-1$ independent parameters where $n=2j+1$ is the dimension of the Hilbert space. Under rotations, the spherical tensor parameters $t^{k}_{q}$ transform elegantly as
$$(t^{k}_{q})^{R} = \sum^{+k}_{q^{'}=-k}\,\, D^{k}_{q^{'}q}(\phi,\theta,\psi)\,t^{k}_{q^{'}},$$
 where $D^{k}_{q^{'}q}(\phi,\theta,\psi)$ is the $(q^{'},q)$ element of the Wigner $D$ matrix, and $(\phi,\theta,\psi)$ are the Euler angles. 
\section{Multiaxial Representation Of Pure and Mixed States}
The spherical tensor parameters $t^k_q$ of a spin-j state possess a geometric representation called the Multiaxial Representation(MAR)\cite{rama}, which is similar to the Majorana Representation. The Majorana representation is applicable to pure symmetric states only whereas MAR is applicable for general mixed spin-$j$ states as well as pure states. The MAR is characterized by the Euler angles $(\theta,\phi,\psi)$ which are related to the parameters $t^k_q$ in the following manner. Consider a rotation $ R(\phi,\theta,0)$ of the frame of reference such that $ t^{k}_{k} $ in the rotated frame vanishes:

$$(t^{k}_{k})^{R} =0= \sum_{q=-k}^{+k}\,\, D^{k}_{qk}(\phi,\theta,0)\,t^{k}_{q}.$$
This implies that using the Wigner expression for $ D^{j} $ matrices\cite{varsh}, we obtain the polynomial equation
\begin{widetext}
\begin{equation}
\chi(\theta,\phi)=\displaystyle \sum_{q=-k}^{k}\,  e^{-iq\phi}\,(-1)^{k-q}\,\sqrt{\binom{2k}{k+q}} t^{k}_{q}\,\left(  cos\frac{\theta}{2}\right)^{k+q} \,(-1)^{k-q}\,\left(  sin\frac{\theta}{2}\right)^{k-q} 
=\mathcal{A}\,\sum_{q=-k}^{+k}\,\sqrt{\binom{2k}{k+q}}\,\,t^{k}_{q}\,Z^{k-q}=0,
\end{equation}
\end{widetext}
where $Z=tan\left( \frac{\theta}{2}\right)\,e^{i\phi}$,
 and the overall coefficient
  \[ \mathcal{A}=cos^{2k}\left( \frac{\theta}{2}\right)\,e^{-ik\phi}\,.\]
A trivial solution is $\theta=\pi$. We therefore redefine $\chi(\cdot,\cdot)$ suitably as a polynomial in $Z$ as
\begin{equation}
P_1(Z)=\displaystyle \sum_{q=-k}^{+k}\sqrt{\binom{2k}{k+q}}\,t^{k}_{q}\,Z^{k-q}=0.
\end{equation} 
Alternatively, it is possible to redefine $\chi(\cdot,\cdot)$ as a polynomial $P_2$ in $Z^{\prime}=\frac{1}{Z}=cot\left( \frac{\theta}{2}\right) \,e^{-i\phi}$ with
\begin{equation}
P_2(Z^{\prime})=\sum_{q=-k}^{+k}\sqrt{\binom{2k}{k+q}}\,t^{k}_{q}\,Z^{\prime^{k+q}}=0,
\end{equation} 
by ignoring the trivial solution $\theta=0$.
In both cases, every $k$ leads to $ 2k $ solutions,
 \[\{(\theta_{1},\phi_{1}),\ldots,(\theta_{k},\phi_{k}),(\pi-\theta_{1},\pi+\phi_{1}),\ldots,(\pi-\theta_{k},\pi+\phi_{k})\}\,.  \]
Thus the $ 2k $ solutions constitute $ k $ axes or $ k $ double headed arrows: for every solution $(\theta_{i},\phi_{i})$, $(\pi-\theta_{i},\pi+\phi_{i})$ also forms a solution. 
The solution set of $P_1$ (equivalently $P_2$) provides the key insight into the geometrical interpretation of the spherical tensor parameters $t^k_q$, elucidated as follows. 
 For a fixed $(\theta_i,\phi_i),i=1,\ldots,k$, consider a unit vector $\hat{Q}_i:=\hat{Q}(\theta_i,\phi_i)$ in $\mathbb{R}^3$. Define 
$$
 s^{k}_{q} = (\ldots((\hat{Q}_1\otimes\hat{Q}_2)^{2}\otimes\hat{Q}_3)^{3}\otimes...\otimes \hat{Q}_{k-1})^{k-1}\otimes\hat{Q}_k)^{k}_{q},
$$
where
$$
 (\hat{Q}_1\otimes\hat{Q}_2)^{2}_{q}=\sum _{q_{1}}C(11k;q_{1}q_{2}q)(\hat{Q}_1)^{1}_{q_{1}} (\hat{Q}_2)^{1}_{q_{2}},
 $$
and the spherical components of $ \hat{Q} $ are given by,
$$
(\hat{Q}(\theta,\phi))^{1}_{q}= \sqrt{\frac{4\pi}{3}}\,\,Y^{1}_{q}(\theta,\phi).
$$
Here $Y^{1}_{q}(\theta,\phi)$ are the well known spherical harmonics.\\

As a consequence, we can state that
$$
 t^{k}_{q} = r_{k}(...((\hat{Q}_1\otimes\hat{Q}_2)^{2}\otimes\hat{Q}_3)^{3}\otimes\cdots\otimes \hat{Q}_{k-1})^{k-1}\otimes\hat{Q}_k)^{k}_{q}. 
$$
Thus in MAR, the symmetric state of $N$-qubit assembly can be represented geometrically by a set of $N=2j$ spheres of radii $r_{1},r_{2},\ldots,r_{k}$ corresponding to each value of $k$. The $k^{th}$ sphere in general consists of a constellation of 2$k$ points on its surface specified by $\hat{Q}_i:=\hat{Q}(\theta_{i},\phi_{i})$ and   $\hat{Q}(\pi-\theta_{i},\pi+\phi_{i}), i=1,2,...,k $. In other words, for a fixed $k$, every $ t^{k}_q, q=-k, -k+1,\ldots,0,1,\ldots,k$, is specified by $k$ axes in a sphere of radius $ r_{k} $. 

\section{Mixed Symmetric Separable States}
Before employing the MAR to develop a criterion for separability, we examine some properties of mixed symmetric separable state. By definition, an N-qubit state is said to be fully separable if it can be decomposed as $\rho = \sum_{i=1}^{n} \lambda_{i}\rho_{i}^{1} \otimes\rho_{i}^{2}\otimes \cdots \otimes\rho_{i}^{N}$, where for $\rho^\alpha_i, \alpha=1,\ldots,N$ is the $i$th decomposition of the system with $\alpha$ as the qubit index. The following Propositions elucidate the relationships between separable, mixed and symmetric states. 
\begin{proposition}
Any N-qubit fully separable state $\rho = \sum_{i=1}^{n} \lambda_{i}\rho_{i}^{1} \otimes\rho_{i}^{2}\otimes \cdots \otimes \rho_i^N$ where $\sum_{i=1}^n\lambda_{i}=1, 0\leq \lambda_i\leq 1$ is mixed if $n\geq2$. 
\end{proposition}
\begin{proof}
Consider the density matrix  for the $\alpha$th qubit where $\alpha=1,\ldots,N$, denoted as $ \rho_{i}^{\alpha}= \frac{1}{2}\left[I+\vec{\sigma}.\cdot\vec{p}_{i}(\alpha)\right], i= 1, 2,...n$, where $\vec{p}_i(\alpha)$ is the polarization vector characterizing the $\alpha$th qubit in the $i$th decomposition. 
For $\rho$ to be pure, $Tr\rho^{2} = 1$ which implies that
\[\sum_{i,j} \lambda_{i}\lambda_{j} Tr(\rho_{i}^{1}\rho_{j}^{1}) Tr (\rho_{i}^{2}\rho_{j}^{2})....Tr(\rho_{i}^{N}\rho_{j}^{N}) = 1.\]
 Therefore \[\sum_{i,j} \lambda_{i}\lambda_{j}[1- Tr(\rho_{i}^{1}\rho_{j}^{1}) Tr (\rho_{i}^{2}\rho_{j}^{2})...Tr (\rho_{i}^{N}\rho_{j}^{N})]= 0.\]
Consequently $ Tr (\rho_{i}^{\alpha}\rho_{j}^{\alpha}) = \frac{1}{2}\left[I + \vec{p}_{i}(\alpha)\cdot\vec{p}_{j}(\alpha)\right]< 1 $
since $\vec{p}_{i}(\alpha)\cdot\vec{p}_{j}(\alpha) < 1$ for $\alpha=1,2...N$, which implies that
\[\sum_{i,j} \lambda_{i}\lambda_{j} (1- Tr(\rho_{i}^{1}\rho_{j}^{1}) Tr (\rho_{i}^{2}\rho_{j}^{2})...  Tr(\rho_{i}^{N}\rho_{j}^{N}))> 0,\]
owing to \[1- Tr(\rho_{i}^{1}\rho_{j}^{1}) Tr (\rho_{i}^{2}\rho_{j}^{2})... Tr (\rho_{i}^{N}\rho_{j}^{N})> 0, \quad \lambda_{i}>0.\]Therefore $$\sum_{i,j} \lambda_{i}\lambda_{j} (1- Tr(\rho_{i}^{1}\rho_{j}^{1}) Tr (\rho_{i}^{2}\rho_{j}^{2})... Tr (\rho_{i}^{N}\rho_{j}^{N})) \not= 0,$$
implying that $\rho$ cannot be pure.
However $\rho$ can be pure if $\vec{p}_{i}(\alpha)\cdot\vec{p}_{j}(\alpha) = 1$  for all $i, j$, in which case there is only one term \[\rho = \rho^{1} \otimes \rho^{2} \otimes .......\otimes \rho^{N}.\]
Therefore, we may define a separable mixed state as \[\rho=\sum_{i=1}^{n}\lambda_{i}\rho_{i}^{1}\otimes\rho_{i}^{2}\otimes....\otimes \rho_{i}^{N}, \]
where $n>1$ and $\rho_{i}^{1},\rho_{i}^{2},...,\rho_{i}^{N} $ are pure. 
\end{proof}
\begin{proposition}
An N-qubit fully separable mixed  state \begin{equation}\rho=\sum_{i=1}^{n}\lambda_{i}\rho_{i}^{1}\otimes\rho_{i}^{2}\otimes....\otimes \rho_{i}^{N}, \end{equation} is permutationally symmetric if $\rho_{i}^{1}=\rho_{i}^{2}=...=\rho_{i}^{N}$.
\end{proposition}
\begin{proof}
Now let us see if symmetrization of two different states $\rho_{i}^{1}$ and $\rho_{i}^{2}$ leads to a state in symmetric subspace.
For some fixed $\lambda_i$ in (5), consider the first two terms $ \rho_{i}^{1} \otimes \rho_{i}^{2} $, and define $$\rho^{12}_i:=\frac{(\rho^{1}_{i} \otimes \rho_{i}^{2}+\rho^{2}_{i} \otimes \rho_{i}^{1})}{2}.$$ Evidently $\rho^{12}_i$ is a density matrix. 

Let $\rho_{i}^{1}= \frac{I+\vec{\sigma}.\vec{p}_{i}(1)}{2}$ and \, $\rho_{i}^{2}= \frac{I+\vec{\sigma}.\vec{p}_{i}(2)}{2}$.  For notational convenience we set $\vec{p}_{i}(\alpha)=\vec{p}_i^\alpha$, $p_{i-}^{2} = p_{ix}^{2}- i p_{iy}^{2}$, $p_{i-}^{1} = p_{ix}^{1}- i p_{iy}^{1}$, $p_{i+}^{2} = p_{ix}^{2}+ i p_{iy}^{2}$, $p_{i+}^{1} = p_{ix}^{1}+ i p_{iy}^{1}$ .
Then
$$\frac{(\rho^{1}_{i} \otimes \rho_{i}^{2}+\rho^{2}_{i} \otimes \rho_{i}^{1})}{2}$$

\begin{widetext}
\begin{displaymath}
=\left[
\begin{array}{cccc}
\frac{(1+p_{iz}^{2})(1+p_{iz}^{1})}{4} &  \frac{p_{i-}^{2}(1+p_{iz}^{1})+p_{i-}^{1}(1+p_{iz}^{2})}{8} & \frac{p_{i-}^{1}(1+p_{iz}^{2})+p_{i-}^{2}(1+p_{iz}^{1})}{8} & \frac{p_{i-}^{1}p_{i-}^{2}}{4}\\ 

\frac{p_{i+}^{2}(1+p_{iz}^{1})+ p_{i+}^{1}(1+p_{iz}^{2})}{8} & \frac{(1-p_{iz}^{2})(1+p_{iz}^{1})+(1-p_{iz}^{1})(1+p_{iz}^{2})}{8} & \frac{p_{i-}^{1}p_{i+}^{2}+p_{i-}^{2}p_{i+}^{1}}{8} & \frac{ p_{i-}^{1}(1-p_{iz}^{2})+p_{i-}^{2}(1-p_{iz}^{1})}{8}\\ 

\frac{p_{i+}^{1}(1+p_{iz}^{2})+p_{i+}^{2}(1+p_{iz}^{1})}{8} & \frac{p_{i+}^{1}p_{i-}^{2}+p_{i+}^{2}p_{i-}^{1}}{8} &  \frac{(1-p_{iz}^{1})(1+p_{iz}^{2})+(1-p_{iz}^{2})(1+p_{iz}^{1})}{8} & \frac{p_{i-}^{2}(1-p_{iz}^{1})+p_{i-}^{1}(1-p_{iZ}^{2})}{8}\\ 

\frac{p_{i+}^{2}p_{i+}^{1}}{4} & \frac{p_{i+}^{1}(1-p_{iz}^{2})+p_{i+}^{2}(1-p_{iz}^{1})}{8} & \frac{(1-p_{iz}^{1})p_{i+}^{2}+(1-p_{iz}^{2})p_{i+}^{1}}{8} & \frac{(1-p_{iz}^{1})(1-p_{iz}^{2})}{4}
\end{array}
\right]
\end{displaymath}\\
\end{widetext} in computational basis {$\ket{\uparrow\uparrow},\ket{\uparrow\downarrow},\ket{\downarrow\uparrow},\ket{\downarrow\downarrow}$}. \\
We can choose a set of basis called angular momentum basis given by $ \{\ket{11}=\ket{\uparrow\uparrow }$, 
 $\ket{10}=\frac{\ket{\uparrow\downarrow}+\ket{\downarrow\uparrow}}{\sqrt 2}, \ket{1-1}=\ket{\downarrow\downarrow}, \ket{00} = \frac{|\uparrow\downarrow\rangle-|\downarrow\uparrow\rangle}{\sqrt 2}\}$ out of which the first three basis states are permutationally symmetric and the last one is permutationally anti-symmetric. The unitary transformation which connects computational basis set to the above set is, 
  \begin{displaymath}
U=\left[
\begin{array}{cccc}
1 & 0 & 0 & 0\\
0 & \frac{1}{\sqrt{2}} & \frac{1}{\sqrt{2}} & 0 \\
0 & 0 & 0 & 1\\
0 & \frac{1}{\sqrt{2}} & \frac{-1}{\sqrt{2}} & 0\\
\end{array}
\right]
.\end{displaymath}
The elements of unitary transformation are the Clebsch-Gordan(CG) coefficients. It is very well known in the angular momentum theory that the Clebsch-Gordan addition of two angular momenta $j_{1}$ and $j_{2}$ resulting in the angular momentum $j$ is given by,
\begin{equation*}
| j_{1} \ j_{2} \ j \ m \rangle = \sum_{m_{1} or m_{2}} C(j_{1} \ j_{2} \ j ; m_{1} \ m_{2} \ m) |j_{1} \ m_{1} \ \rangle |j_{2} \ m_{2} \rangle . 
\end{equation*}
Thus, 2-qubit symmetric state has one-to-one correspondence with a spin-$1$ state and the most general 2-qubit state resides in $2^{2}= 4$ dimensional Hilbert space. Thus, the Clebsch-Gordan decomposition of the space is given by $2\otimes 2 = 3 \oplus 1 $, where the highest, that is the $3$-dimensional space is the symmetric subspace.
Now, transforming the density matrix to symmetric subspace , we get \\ \\ \\

\begin{widetext}
\begin{displaymath}
U \rho U^{\dagger}=\left[
\begin{array}{cccc}
\frac{(1+p_{iz}^{2})(1+p_{iz}^{1})}{4} & \frac{p_{i-}^{1}(1+p_{iz}^{1})+p_{i-}^{1}(1+p_{iz}^{2})}{4\sqrt{2}} & \frac{p_{i-}^{1}p_{i-}^{2}}{4} & 0\\ 

\frac{p_{i+}^{2}(1+p_{iz}^{1})+p_{i+}^{1}(1+p_{iz}^{2})}{4\sqrt{2}} & \frac{(1-p_{iz}^{2}p_{iz}^{1}+p_{ix}^{1}p_{ix}^{2}+p_{iy}^{1}p_{iy}^{2})}{4} & \frac{p_{i-}^{1}(1-p_{iz}^{2})+ p_{i-}^{2}(1-p_{iz}^{1})}{4\sqrt{2}} & 0\\ 

\frac{p_{i+}^{1}p_{i+}^{2}}{4} & \frac{p_{i+}^{2}(1-p_{iz}^{1})+p_{i+}^{1}(1-p_{iz}^{2})}{4\sqrt{2}} &  \frac{(1-p_{iz}^{2})(1-p_{iz}^{1})}{4} & 0\\ 

0 & 0 & 0 & \frac{1-\hat{p_{i}^{2}}.\hat{p_{i}^{1}}}{4}
\end{array}
\right]
\end{displaymath}\\
\end{widetext}
Thus symmetrization of $\rho^{1}_{i}$ and $\rho^{2}_{i}$ does not lead to a state in symmetric subspace unless $\frac{1-\hat{p_{i}^{2}}.\hat{p_{i}^{1}}}{4} = 0$. This implies that
$\hat{p_{i}^{2}}=\hat{p_{i}^{1}}$. 
Similarly by continuing in the same way, considering the permutational symmetry of all N-qubit taking two qubit at a time, we get $\hat{p_{i}^{1}}=\hat{p_{i}^{2}} = \hat{p_{i}^{3}}= .... = \hat{p_{i}^{N}}$
Thus all the N-qubits in a partition will have same vector polarization.
Therefore a mixed symmetric separable state is an ensemble of symmetric pure separable states and henceforth we write it as $ \rho = \sum_{i}\lambda_{i}\rho_{i} \otimes \rho_{i}\otimes....\otimes\rho_{i}.$ 
\end{proof}

\section{Multiaxial Representation of Mixed Symmetric Separable States.}

To arrive at the multiaxial representation of an N-qubit symmetric separable state, let us consider $$ \rho = \sum_{i}^{n}\lambda_{i}\rho_{i} \otimes \rho_{i}\otimes \ldots \otimes \rho_{i}= \sum_{i}^{n}  \lambda_{i}  \varrho_{i}^{N} $$ where $\varrho_{i}^{N} = \rho_{i} \otimes \rho_{i} \otimes \ldots \otimes \rho_{i}$.

  The unitary transformation $\mathcal{U}$, decomposes $\varrho_{i}^{N}$ into the direct sum of its composite density matrices out of which $2j+1$ or $2(\frac{N}{2})+1$ dimensional density matrix $\rho_{i}^{'}$ is totally symmetric and the rest are zeros. The elements of $\mathcal{U}$ are the well known CG coefficients. Therefore, $\rho$ in symmetric subspace is written as \begin{equation}\rho^{j}_{symm} = \sum_{i} \lambda_{i} \rho_{i}^{'}.\end{equation}  From equation (1),

$$\rho^{j}_{symm}=\frac{1}{2j+1} \sum_{kq} t_{q}^{k} \tau_{q}^{k^{\dagger}} = \frac{1}{2j+1} \sum_{i}\sum_{kq} \lambda_{i} t_{q}^{k} (i) \tau_{q}^{k^{\dagger}}$$ which implies that,
\begin{equation}
t_{q}^{k} = \sum_{i}^{n} \lambda_{i} t_{q}^{k} (i).
\end{equation}
Clearly, each of the $\rho_{i}^{'}$'s is a pure spin-j density matrix expressed in the $\ket{jm}$ basis and the corresponding density matrix is $(\rho_{i} \otimes \rho_{i}\otimes...\otimes\rho_{i})$ in the computational basis which can also be written as $ \ket{\psi_{i}\psi_{i}\cdots\psi_{i}}\bra{\psi_{i}\psi_{i}\cdots\psi_{i}}$. The MAR of pure symmetric separable states has already been investigated\cite{samin} which we introduce here briefly: As $\rho_{i}^{'}$ is characterized by $(\theta_{i},\phi_{i})$, in a rotated frame of reference, whose z-axis is parallel to $(\theta_{i},\phi_{i})$, $\rho_{i}^{'}$ assumes a canonical form given by 
\begin{equation}
 \left(
\begin{matrix}
1 & 0 & \ldots & 0 \\
0&0 & \ldots &0 \\
\vdots & \vdots & \ddots & \vdots \\
0 & 0 & \ldots & 0\\
\end{matrix}
\right),
\end{equation} in angular momentum basis.
The only non-zero spherical tensor parameters characterizing the above state are, 
$$({t^{k}_{0}})^{rotated}=[k]\,\rho_{jj}\,C(jkj;j0j),\,\,\,\,\,\,\,\,\,\,\, k=0,1,2...2j$$
 and thus, each $t_{q}^{k} (i)$ is constructed out of single axis and the resultant $\rho_{i}^{'}$'s is characterized by one axis, namely $\hat{Q}(\theta_{i},\phi_{i})$ and represents a uniaxial system. In other words $t_{q}^{k} (i) \propto Y_{q}^{k}(\theta_{i}\phi_{i})$. Hence we write (7) as

$$t_{q}^{k} = C \sum_{i} \lambda_{i} Y_{q}^{k} (\theta_{i}\phi_{i}).$$ where C is a proportionality constant. In the continuum limit it is natural to take $t^{k}_{q}$ as,
\begin{equation}
t_{q}^{k} = C\int \lambda(\theta,\phi) Y_{q}^{k} (\theta,\phi) d\Omega
\end{equation}
where $\int \lambda(\theta,\phi) d\Omega = 1 $, $\lambda(\theta,\phi)$ is positive and $d\Omega = \sin{\theta} d\theta d\phi$.\\

Now, it is interesting to investigate the functional form of $t^{k}_{q}$ in the continuum limit. To do this, let us consider the angular momentum operator  $L^{2} = \vec{L}.\vec{L}$ and $L_{z}$ which have the form $$L_{z} = -i\hbar \frac{\partial}{\partial \phi}$$ $$L^{2} = -\hbar^{2} \frac{1}{\sin\theta}\frac{\partial}{\partial\theta}(\sin\theta \frac{\partial}{\partial\theta})-\hbar^{2}\frac{1}{\sin^{2}\theta}\frac{\partial^{2}}{\partial \phi^{2}}.$$ It can be easily seen that that $t_{q}^{k}$ is a simultaneous eigen state of $L^{2}$ and $L_{z}$ if $\lambda$ is independent of $\theta$ and $\phi$. In such a case, for every $k$, the $t^{k}_{q}$'s of the mixed symmetric separable state are characterized by k axes which are collinear.\\

\subsection{Classicality and Separability}
It is well known that a density matrix $\rho$ is called P-representable(P-rep) if it can be
written as a convex sum of coherent states, $\rho$ = $\int d{\alpha} P({\alpha}) \ket{\alpha}\bra{\alpha}$ where $\alpha$ is a coherent state and $P(\alpha)$ is a probability density function with $\int P(\alpha) d\alpha = 1$. We can identify our pure separable symmetric state $\rho_{i}^{'}$ of (6) with the coherent states as coherent states are the rotated $\ket{jj}$ states\cite{coherent} i.e,
$$ \ket{\alpha(\theta,\phi)} = \sum_{m} \ket{jm}\bra{jm}R(\phi,\theta,0)\ket{jj} = \\
\sum_{m} D^{j}_{mj}(\phi,\theta,0)\ket{jm}$$
$$ =\displaystyle \sum_{m=-j}^{j} \sqrt{\binom{2j}{j+m}} (\sin\theta)^{j-m} (\cos\theta)^{j+m} e^{-i(j+m)\phi} \ket{jm}$$ $$ $$ where $D^{j}_{mj}(\phi,\theta,0)$ is wigner D matrices.\\

 Rotated $\ket{jj}$ states assume canonical form as shown in equation(8), and hence they are pure separable states. Thus, N-qubit mixed symmetric separable states are identified with P-rep states.
 Any density matrix which is P-rep is widely accepted as classical state\cite{prep1,prep2}. Hence N-qubit symmetric separable states are classical spin-$\frac{N}{2}$ states.\\
 Conversely, a classical spin-j state with $j=\frac{N}{2}$ which is a convex mixture of coherent states can be realized as an N-qubit symmetric separable state in $2j$ tensor product space, proof of which is given in \cite{ppt}.  \\It has already been proved that $P(\alpha)$ is not uniquely determined by the density operator\cite{classicality}. We choose the most general positive function on the sphere, $\lambda(\theta,\phi) $ as the $P$ function\cite{Pfun} and study the MAR of the corresponding density matrix.
 If
  $ \lambda(\theta,\phi) = \displaystyle  \sum_{l=0}^{\infty}\sum_{m=-l}^{l} a^{l}_{m} Y^{l^{*}}_{m}(\theta,\phi)$
  then, $$ t^{k}_{q} = \int \displaystyle  \sum_{l=0}^{\infty}\sum_{m=-l}^{l} a^{l}_{m} Y^{l^{*}}_{m}(\theta,\phi) Y^{k}_{q}(\theta,\phi) d\Omega $$ $$ = \sum_{lm} a^{l}_{m} \delta_{kl} \delta_{qm} =  a^{k}_{q}. $$
  Thus, $t^{k}_{q}$'s are the expansion coefficients of the probability density function $\lambda(\theta,\phi)$. Given $a^{k}_{q}$'s, we can explicitly determine the axes from MAR. 
  
  Now as an example, let us choose a probability density function of the form $Y^{l}_{m}(\theta,\phi) {Y^{l}_{m}}^{*}(\theta,\phi)$ and study the MAR of the $t^{k}_{q}$'s belonging to the N-qubit mixed symmetric separable state. i.e,\\
$$ t_{q}^{k} = C\int  Y^{l}_{m}(\theta,\phi) {Y^{l}_{m}}^{*}(\theta,\phi) Y_{q}^{k} (\theta,\phi)d\Omega.$$
Using equation (11) of section (5.6) of \cite{varsh},
\begin{widetext}
$$t_{q}^{k}= C \int \sum_{LL^{'}} (-1)^{m} \sqrt{\frac{(2l+1)^{2}(2k+1)}{(4\pi)^{2} (2L+1)}} C(llL^{'}:000) C(L^{'}kL:000) C(llL^{'}:m -m 0) C(L^{'}kL:0qq) Y^{L}_{q}(\theta,\phi) d\Omega$$
\end{widetext}
and after integration,
\begin{widetext}
$$t_{q}^{k}= C  \sum_{LL^{'}} (-1)^{m} \sqrt{\frac{(2l+1)(2l+1)(2k+1)}{(4\pi)^{2}(2L+1)}} C(llL^{'}:000) C(L^{'}kL:000) C(llL^{'}:m -m 0) C(L^{'}kL:0qq) \delta_{L0} \delta_{q0} \sqrt{4\pi}.$$
\end{widetext}
Therefore the only non-zero $t^{k}_{q}$'s are given by
\begin{widetext}
$$t^{k}_{0}= C  \sum_{L^{'}} (-1)^{m} \sqrt{\frac{(2l+1)(2l+1)(2k+1)}{(4\pi)^{2}(2L+1)}} C(llL^{'}:000) C(L^{'}k0:000) C(llL^{'}:m -m 0) C(L^{'}k0:0qq) \sqrt{4\pi}.$$
\end{widetext}
Thus, the N-qubit mixed separable symmetric state $\rho$ characterized by the above $t^{k}_{0}$'s is a uniaxial system with the axes being collinear to the z-axis as explained in the section V(see Fig.1).

 \begin{figure}
  \includegraphics[width=\linewidth]{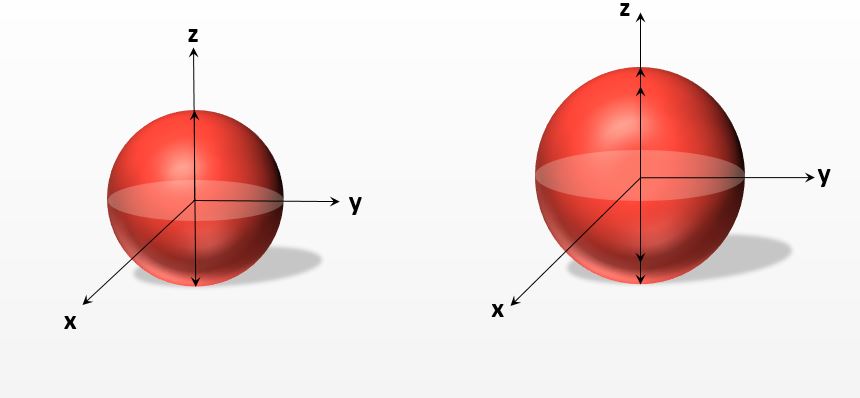}
  {Fig1 : MAR of mixed symmetric separable state showing the axes characterized by $t^{k}_{0}$ for $k=1,2.$}
  
\end{figure}

Now let us take an example of a discrete case; 2-qubit mixed symmetric separable state with uniform distribution as follows;
$$\rho=\frac{1}{4} (\rho_{x} \otimes \rho_{x})+\frac{1}{4} (\rho_{-x} \otimes \rho_{-x})+\frac{1}{4} (\rho_{z} \otimes \rho_{z})+ \frac{1}{4} (\rho_{-z} \otimes \rho_{-z})$$ where x, -x, z, -z are four maximally separated points on bloch sphere and $\rho_{i} = \frac{I+\sum_{i}\sigma_{i}p_{i}}{2}$ ; i= x,y,z. Explicitly,
 \begin{displaymath}
\rho=\frac{1}{16}\left[
\begin{array}{cccc}
6 & 0 & 0 & 2\\
0 & 2 & 2 & 0 \\
0 & 2 & 2 & 0\\
2 & 0 & 0 & 6\\
\end{array}
\right]
\end{displaymath}
 and transforming $\rho$ to $|jm>$ basis, we have
\begin{displaymath}
\rho^{jm}=\frac{1}{16}\left[
\begin{array}{ccc}
6 & 0 & 2\\
0 & 4 & 0\\
2 & 0 & 6\\
\end{array}
\right]
\end{displaymath}
The non-zero $t^{k}_{q}$'s are $t^{2}_{0} = \frac{1}{4\sqrt{2}}$, $t^{2}_{2} = t^{2}_{-2} = \frac{\sqrt{3}}{8}$. The polynomial equation (3) for $k=2$ becomes $$ z^{4} \frac{\sqrt{3}}{8}+ z^{2} \frac{\sqrt{6}}{4\sqrt{2}}+\frac{\sqrt{3}}{8}=0$$ solutions of which give us the two collinear axes namely, $\hat{Q}(\frac{\pi}{2},\frac{\pi}{2})$ and $\hat{Q}(\frac{\pi}{2},\frac{\pi}{2})$.\\

 Therefore, in each of the above cases, $t_{q}^{k}$ $\propto $ $Y_{q}^{k}(\theta,\phi).$
 Now that we have given a MAR for mixed symmetric separable state or equivalently for classical states, one can explore the quantumness of such states analytically.

  
 
\section{Conclusion}
We have identified the mixed symmetric fully separable N-qubit state with spin-j density matrix and expressed it in terms of Fano statistical tensor parameters. Using the Multiaxial Representation of density matrix, we realize that, a fully separable N-qubit symmetric state is characterized by spherical tensor parameters $t^{k}_{q}$'s which are always proportional to spherical harmonics $Y^{k}_{q}(\theta,\phi)$ in the continuum limit when the $P$ distribution function $\lambda(\theta,\phi)$ is independent of $\theta$ and $\phi$.  Further it is shown that for such a case  the mixed symmetric separable states are characterized by collinear axes. In contrast, for a general density matrix each $t^{k}_{q}$ is characterized by $k$ distinct axes. We have also identified  N-qubit mixed symmetric separable states with P-rep states or classical states. Since the distribution function is not uniquely decided by the density matrix, we have chosen it to be the most general positive function on a sphere of unit radius and concluded that $t^{k}_{q}$'s are given by expansion coefficients of the $P$ function. By chosing $Y^{l}_{m} Y^{l^{*}}_{m}$ as the probability density function we have proved that the corresponding state is characterized by non zero $t^{k}_{0}$'s only. In other words the axes are collinear.  We have also examined the MAR of a two qubit mixed symmetric state consisting of four terms with equal weightage and concluded that it is characterized by collinear axes.   \\ \\

\appendix*

\acknowledgments
One of the authors, SH thanks Department of Science and Technology(DST,India) for financial assistance through its scholarship for Higher Education(INSPIRE) programme.

\bibliography{ms}

\end{document}